\documentclass[runningheads]{llncs}
\usepackage{graphicx,color}
\usepackage{hyperref}
\usepackage{amssymb,amsmath,cleveref,enumerate}
\usepackage[all]{xy}
\newcommand{\mt}{\mathtt}

\newcommand{\abs}[1]{\lvert#1\rvert}

\begin{document}
\title{Conditional automatic complexity\\ and its metrics\thanks{This work was partially supported by a grant from the Simons Foundation (\#704836 to Bj\o rn Kjos-Hanssen)}}
\author{Bj{\o}rn Kjos-Hanssen\inst{1}\orcidID{0000-0002-6199-1755}}
\authorrunning{B. Kjos-Hanssen}%
\institute{University of Hawai\textquoteleft i at M\=anoa, Honolulu HI 96822, USA\\
\email{bjoernkh@hawaii.edu}\\
\url{https://math.hawaii.edu/wordpress/bjoern/}}
\maketitle
\begin{abstract}
Li, Chen, Li, Ma, and Vit\'anyi (2004) introduced a similarity metric based on Kolmogorov complexity. It followed work by Shannon in the 1950s on a metric based on entropy. We define two computable similarity metrics, analogous to the Jaccard distance and Normalized Information Distance, based on conditional automatic complexity and show that they satisfy all axioms of metric spaces.
\end{abstract}
\section{Introduction}

	In this article we show that metrics analogous to the Jaccard distance and the Normalized Information Distance can be defined based on conditional nondeterministic automatic complexity $A_N$.
	Our work continues the path of Shannon (1950) on entropy metrics and G\'acs (1974) on symmetry of information among others.

	Shallit and Wang (2001) defined the automatic complexity of a word $w$ as, somewhat roughly speaking, the minimum number of states of a finite automaton that accepts $w$ and no other word of length $\abs{w}$. This definition may sound a bit artificial, as it is not clear the length of $w$ is involved in defining the complexity of $w$. In this article we shall see how \emph{conditional} automatic complexity neatly resolves this issue.

	\begin{definition}[{\cite{MR3386523,MR1897300}}]\label{precise}
		Let $L(M)$ be the language recognized by the automaton $M$.
		Let $x$ be a sequence of finite length $n$.
		The (unique-acceptance) nondeterministic automatic complexity
		$A_N(w)=A_{Nu}(w)$ of a word $w$ is
		the minimum number of states of an NFA $M$
		such that $M$ accepts $w$ and the number of walks
		along which $M$ accepts words of length $\abs{w}$ is 1.

		The exact-acceptance nondeterministic automatic complexity
		$A_{Ne}(w)$\index{$A_{Ne}(x)$} of a word $w$ is
		the minimum number of states of an NFA $M$
		such that $M$ accepts $w$ and $L(M)\cap\Sigma^{\abs{w}}=\{w\}$.

		The (deterministic) automatic complexity $A(w)$ is the minimum number of states of a DFA $M$ with $L(M)\cap\Sigma^{\abs{w}}=\{w\}$.

		Finally, $A^-(w)$ is the minimum number of states of a deterministic but not necessarily complete (total) NFA with $L(M)\cap\Sigma^{\abs{w}}=\{w\}$.
	\end{definition}

	\begin{remark}
		$A^-(w)$ is so named because it satisfies $A^-(w)\le A(w)\le A^-(w)+1$.
	\end{remark}

	\begin{lemma}\label{lem:cocoon-referee}
		For each word $w$, we have $A_{Ne}(w)\le A_{Nu}(w)$.
	\end{lemma}
	\begin{proof}
		We simply note that if $M$ uniquely accepts $w$, then $M$ exactly accepts $w$.
	\end{proof}

	\begin{remark}\label{cocoon-referee}
		Let $w=00$ and let $M$ be the following NFA:
		
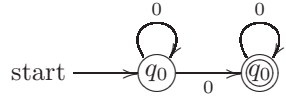
\[
		\xymatrix{
		\text{start}\ar[r] & *+[Fo]{q_0}\ar[r]_0\ar@(ul,ur)^0 & *+[Foo]{q_0}\ar@(ul,ur)^0\\
		}
		\]
		We see that $M$ accepts $w$ on two distinct walks, $(q_0,q_0,q_1)$ and $(q_0,q_1,q_1)$. Hence $M$ exactly accepts $w$, but does not uniquely accept $w$.
		This example is contrived in the sense that we could remove an edge without changing $L(M)$, but it is not clear that this will always be possible for other NFAs.
		Thus the question whether $A_{Nu}=A_{Ne}$, considered in \cite{MR3938583}, remains open.
	\end{remark}

\section{Conditional complexity}
	\begin{definition}[Track of two words]
		Let $\Gamma$ and $\Delta$ be alphabets.
		Let $n\in\mathbb N$, $x\in\Gamma^n$ and $y\in\Delta^n$.
		When no confusion with binomial coefficients is likely, we let $\binom{a}{b}=(a,b)\in\Sigma\times\Delta$.
		The \emph{track} of $x$ and $y$, $x\# y\in (\Gamma\times\Delta)^n$, is defined to be the word
		\[
			\binom{x_0}{y_0} \binom{x_1}{y_1}\dots \binom{x_{n-1}}{y_{n-1}},
		\]
		 which we may also denote as $\binom{x}y$.
	\end{definition}
	\begin{definition}[Projections of a word]
		Let $\Gamma$ and $\Delta$ be alphabets.
		Let $n\in\mathbb N$, $x\in\Gamma^n$ and $y\in\Delta^n$.
		The projections $\pi_1$ and $\pi_2$ are defined by
		$\pi_1(x\# y)=x$, $\pi_2(x\# y)=y$.
	\end{definition}

	$x\# y$ can be thoughts of as a parametrized curve $i \mapsto (x_i,y_i)$. The symbol $\#$ reminds us of the two ``tracks'' corresponding to $x$ and $y$.
	This use of the word ``track'' can be found in \cite{Shallit:2008:SCF:1434864}.

	\begin{theorem}\label{jun22-23}
		There exist words $x, y$ with $A_N(x\# y)\not\le A_N(x)+A_N(y)$.\footnote{Recall that $A_N:=A_{Nu}$.}
	\end{theorem}
	\begin{proof}
		Let $x=(010)^4$ and $y=(01)^6$. Then we check that $A_N(x\#y)=6$, $A_N(x)=3$, and $A_N(y)=2$.
	\end{proof}
	\begin{definition}\label{word-as-permutation}
		A permutation word is a word that does not contain two occurrences of the same symbol.
	\end{definition}

	\begin{theorem}
		For all words $x, y$, we have $\max\{A_N(x),A_N(y)\}\le A_N(x\# y)$.
		There exist words $x, y$ with $\max\{A^-(x),A^-(y)\}\not\le A^-(x\# y)$.
	\end{theorem}
	\begin{proof}
		Let $x$ be a word of some length $n$ with $A^-(x)>n/2+1$.
		An example can be found among the maximum-length sequences for linear feedback shift registers as observed in \cite{MR3828751}.
		Let $y$ be a permutation word (\Cref{word-as-permutation}) of the same length.
		Whenever $y$ is a permutation word, so is $x\# y$. Therefore $A^-(x\# y)\le n/2+1< A^-(x)$.
	\end{proof}

	\begin{definition}\label{conditional-automatic-complexity}
		Let $\Gamma$ and $\Delta$ be alphabets. Let $n\in\mathbb N$ and $x\in\Gamma^n, y\in\Delta^n$.
		The \emph{conditional (nondeterministic) automatic complexity of $x$ given $y$}, $A_N(x\mid y)$, is
		the minimum number of states of an NFA over $\Gamma\times\Delta$ such that \Cref{cond-one} and \Cref{cond-two} hold.
		\begin{enumerate}[(i)]
			\item\label{cond-one} Let $m$ be the number of accepting walks of length $n=\abs{x}=\abs{y}$ for which the word $w$ read on the walk satisfies $\pi_1(w)=y$.
			Then $m=1$.
			\item\label{cond-two} Let $w$ be the word in \Cref{cond-one}. Then $\pi_2(w)=x$.
		\end{enumerate}
	\end{definition}
	An example of \Cref{conditional-automatic-complexity} is given in \Cref{cocoon-ref-example}.

	The conditional complexity $A(x\mid y)$ must be defined in terms of a unique sequence of edges rather than a unique sequence of states,
	since we cannot assume that there is only one edge from a given state $q$ to given state $q'$.

	\begin{theorem}\label{jun16-23}\label{jun17-23}
		$A_N(x\#y)\le A_N(x\mid y)\cdot A_N(y)$.
		In relativized form, $A_N(x\mid z) \le A_N(y\mid z) \cdot A_N(x\mid y\#z)$.
	\end{theorem}
	\begin{proof}
		We describe the unrelativized form only.
		We use a certain product of NFAs.\footnote{This construction may well have appeared elsewhere but we are not aware of it.}
		Let two NFAs
		\[
			M_1=(Q_1,\Gamma\times\Delta,\delta_1,q_{0,1},F_1),\quad
			M_2=(Q_2,\Gamma, \delta_2, q_{0,2},F_2)
		\]
		be given.
		The product is $M_1\times_1 M_2=(Q_1\times Q_2, \Delta, \delta, (q_{0,1},q_{0,2}), F_1\times F_2)$ where
		$\delta((q, q'),a)\ni (r, r')$ if $\delta_1(q,(b,a))\ni r$ and $\delta_2(q',b)\ni r'$ for some $b$.

		(We can also form $M_1\times_2 M_2$ where $\delta((q, q'),(b,a))\ni (r, r')$ if $\delta_1(q,(b,a))\ni r$ and $\delta_2(q',b)\ni r'$.)

		For a walk $w$, let $\mathrm{word}(w)$ be the word read on the labels of $w$.
		Consider an accepting walk $w$ from $(q_1,q_2)$ to $(r_1,r_2)$.
		By definition of the start and final states of $M_1\times_1 M_2$, the projection $\pi_1(w)$ is also accepting.
		Hence by \emph{the $A_N(y)$ witness assumption} $\pi_1(w)$ is the only accepting walk of its length and $\mathrm{word}(\pi_1(w))=y$.
		Since $\mathrm{word}(\pi_1(w))=y$, by \emph{the $A_M(x\mid y)$ witness assumption} $w$ is the unique walk with $\mathrm{word}(\pi_1(w))=y$,
		and $\mathrm{word}(\pi_2(w))=x$.
		Thus the accepted word is $x\# y$ and $w$ is the unique accepting walk of its length.
	\end{proof}

	\begin{theorem}\label{cor-jun16-23}
		There exist $x$ and $y$ with $A_N(x\#y) \ne A_N(x\mid y)\cdot A_N(y)$.
	\end{theorem}
	\begin{proof}
		Let $y=\mt{0001}$ and $x=\mt{0123}$. It is enough to note that $A_N(x\#y)=3$, $A_N(y)=2$, and $A_N(x\mid y)$ is an integer.
	\end{proof}

	\begin{remark}\label{cocoon-ref-example}
		\Cref{jun16-23} is not optimal in the sense of \Cref{cor-jun16-23}.
		On the other hand, \Cref{jun16-23} is optimal in the sense that there is a class of word pairs for which it cannot be improved:
		let $y=(\mt{012345})^k$ for some large $k$, and let $x=(\mt{0123})^l$ where $4l=6k$, so that $\abs{x}=\abs{y}$.
		We have $A_N(x\# y)=\mathrm{lcm}(4,6)=12$, $A_N(y)=6$, and $A_N(x\mid y)=2$ as witnessed by the NFA in \eqref{wit-by}.
		\begin{equation}\label{wit-by}
			\begin{gathered}
				\xymatrix{
				M_1&	 \text{start}\ar[d]\\
				&	 *+[Foo]{q_0} \ar@/_/[r]_{\binom51}\ar@(dl,ul)^{\binom00, \binom11, \binom22, \binom33,\binom40}
					 &
					 *+[Fo]{q_1}  \ar@/_/[l]_{\binom53}\ar@(dr,ur)_{\binom02, \binom13, \binom20, \binom31,\binom42}
				}
			\end{gathered}
		\end{equation}
		The product of this and a cyclic $M_2$ automaton for $y$ is another cyclic automaton, shown in \eqref{wit-2}.
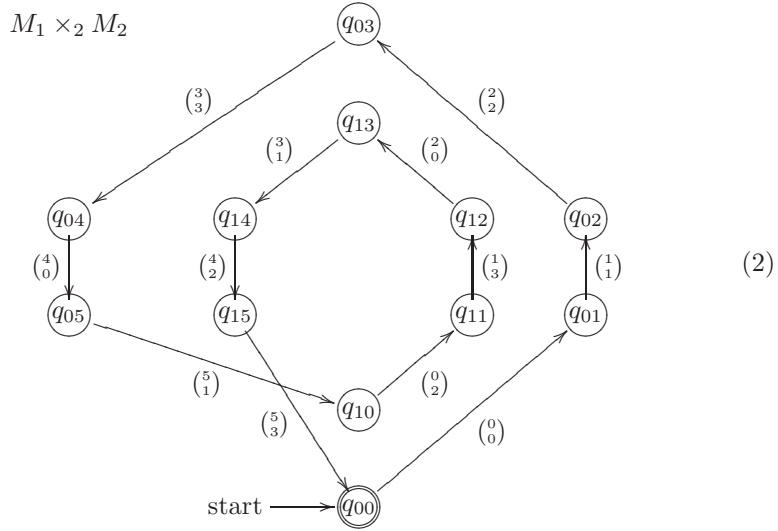
		\begin{equation}\label{wit-2}
			\begin{gathered}
				\xymatrix{
M_1\times_2 M_2						&									&*+[Fo]{q_{03}}\ar[ddll]_{\binom33}	& 									&	\\
									&									&*+[Fo]{q_{13}}\ar[dl]_{\binom31}	&									&	\\
*+[Fo]{q_{04}}\ar[d]_{\binom40}		&*+[Fo]{q_{14}}\ar[d]_{\binom42}	&									&*+[Fo]{q_{12}}\ar[ul]_{\binom20}	&	*+[Fo]{q_{02}}\ar[uull]_{\binom22}\\
*+[Fo]{q_{05}}\ar[drr]_{\binom51}	&*+[Fo]{q_{15}}\ar[ddr]_{\binom53}	&									&*+[Fo]{q_{11}}\ar[u]_{\binom13}	&	*+[Fo]{q_{01}}\ar[u]_{\binom11}\\
									&									&*+[Fo]{q_{10}}\ar[ur]_{\binom02}	&									&	\\
									&\text{start}\ar[r]					&*+[Foo]{q_{00}}\ar[uurr]_{\binom00}& 									&	\\		
			}   
			\end{gathered}	
		\end{equation}
	\end{remark}

	\begin{definition}
		A word $x$ induces an equivalence relation $i\sim_x j\iff x_i=x_j$.
	\end{definition}
	\begin{theorem}
		If $\sim_x$ refines $\sim_y$ then $A_N(y\mid x)=1$.

		If $y$ is a constant word then $A_N(x\mid y)=A_N(x)$.
	\end{theorem}

	An equivalent characterization is that $A(x\mid y)$ is the minimum number of states of an NFA that accepts $y$ on only one walk (but may accept other words of the same length),
	such that the equivalence relation induced by the sequence of labeled edges used refines $\sim_x$.

\section{Bearing on the unique vs.~exact problem}
	A central problem in automatic complexity is whether $A_{Ne}=A_{Nu}$ \cite{MR3938583}. Moving to conditional complexity sheds new light.

	\begin{definition}
		A \emph{sparse witness} for $A_{Ne}(x\mid y)$ is an NFA $M$ that witnesses the value of $A_{Ne}(x\mid y)$, with the additional properties that
		\begin{enumerate}[(i)]
			\item if any edge is removed from $M$, then it is no longer a witness of $A_{Ne}(x\mid y)$; and
			\item $M$ has fewer or equal number of edges as some witness $M_1$ of $A_{Nu}(x\mid y)$.
		\end{enumerate} 
	\end{definition}
	Note that if $A_{Ne}(x\mid y)=A_{Nu}(x\mid y)$ then any witness for $A_{Nu}(x\mid y)$ is a sparse witness for $A_{Ne}(x\mid y)$. The converse fails:
	\begin{theorem}
		There exist binary words $x,y$ such that there is a $A_{Ne}(x\mid y)$ sparse witness that is not an $A_{Nu}(x\mid y)$-witness.
	\end{theorem}
	\begin{proof}
		We will display slightly more than promised: an $A_{Ne}(x\mid y)$ witness that has \emph{strictly fewer} edges than some $A_{Nu}(x\mid y)$ witness for the same $x,y$.
		Consider $x=0000110$, $y=0010100$, and the state sequences $00111200$ and
		$01111200$. They both generate the same NFA:
		
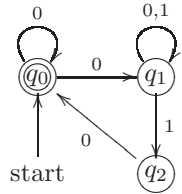
\[
		\xymatrix{
		*+[Foo]{q_0}\ar@(ul,ur)^0\ar[r]^0 & *+[Fo]{q_1}\ar@(ul,ur)^{0,1}\ar[d]^1\\
		\text{start}\ar[u]&*+[Fo]{q_2}\ar[ul]^0
		}
		\]
		They are sparse witnesses, and have only 6 edges, whereas an $A_{Nu}$ witness with 7 edges is the state sequence 01200210.
		
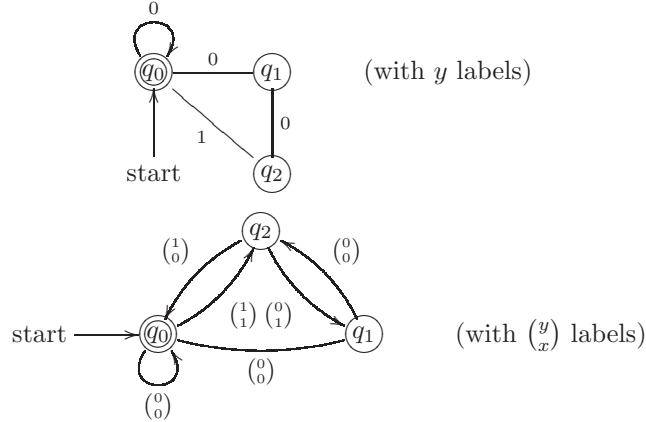
\[
		\xymatrix{
		*+[Foo]{q_0}\ar@(ul,ur)^0\ar@{-}[r]^0 & *+[Fo]{q_1}\ar@{-}[d]^0&\text{(with $y$ labels)}\\
		\text{start}\ar[u]&*+[Fo]{q_2}\ar@{-}[ul]^1
		}
		\]
		\[
			\xymatrix{
				&
				&
				*+[Fo]{q_2}\ar@/_/[dl]_{\binom10}\ar@/_/[dr]_{\binom01}\\
				\text{start}\ar[r]
				&
				*+[Foo]{q_0}
				\ar@/_/[ur]_{\binom11}
				\ar@(dl,dr)_{\binom00}
				\ar@{-}@/_/[rr]_{\binom00}
				&
				&
				*+[Fo]{q_1}
				\ar@/_/[ul]_{\binom00}
				&
				\text{(with $\binom{y}x$ labels)}\\
			}
		\]
		Here, the arrowless edges represent two edges with the same label in opposite directions.
	\end{proof}

	We do not know if sparse witnesses can ever be found in the unconditional case $A_{Ne}(x)$.
	For example, for $x=01110$ the state sequence is 011220 is a non-sparse witness; and there are no sparse witnesses for any $\abs{x}\le 8$.

\section{A Jaccard distance metric}
	For binary words $x,y$ let
	\begin{equation}\label{eq:metric}
		J^{\mathrm{num}}(x,y)=\log (A_N(x\mid y)A_N(y\mid x)).
	\end{equation}
	The base of logarithm chosen is not important, but it is sometimes convenient to let $\log=\log_2$.

	Let us briefly recall the definitions of metric and pseudometric spaces.
	\begin{definition}
		Let $X$ be a set.
		A function $d:X\times X\to\mathbb R_{\ge 0}$ is a \emph{pseudometric} if it is
		commutative, satisfies the triangle inequality, and satisfies $d(x,x)=0$ for all $x\in X$.
		If in addition $d(x,y)=0\implies x=y$ then $d$ is a \emph{metric}.
	\end{definition}
	The underlying space for our metrics will be $\alpha^n/\mathrm{Sym}(\alpha)$ where
	$\alpha$ is an alphabet, $n\in\mathbb N$, and $\mathrm{Sym}(\alpha)$ is the symmetric group of all permutations of $\alpha$.
	Our metrics arise from pseudometrics on $\alpha^n$.
	If $\alpha$ is finite we can assume $\alpha=[a]=\{0,1,\dots,a-1\}$ for some $a\in\mathbb N$, and choose \emph{slow} sequences as representatives.
	Here, a sequence $s\in\alpha^n$ is slow if $s(0)=0$ and for each $i$, $s(i)\le\max\{s(j)\mid j<i\}+1$.
	In the case of a binary alphabet $\{0,1\}$, this set of representatives is simply $0\,\{0,1\}^{n-1}$.

	\begin{theorem}\label{jun26-23}
		$J^{\mathrm{num}}$ is a metric on the set $0\{0,1\}^{n-1}$ for any $n\ge 1$. 
	\end{theorem}
	\begin{proof}
		We have $J^{\mathrm{num}}(x,y)=0$ iff $x$ and $y$ are isomorphic under some permutation of $\Sigma$.
		If we restrict attention to binary words of the form $0z$ we get $J^{\mathrm{num}}(x,y)=0\iff x=y$.
		And $J^{\mathrm{num}}(x,y)=J^{\mathrm{num}}(y,x)$ is immediate.
		\Cref{jun17-23} implies
		\[
			A_N(x\mid y) \le A_N(z\mid y) \cdot A_N(x\mid z),
		\]
		hence
		\begin{eqnarray*}
			J^{\mathrm{num}}(x,y) &=& \log_2 (A_N(x\mid y)A_N(y\mid x))\\
			&\le& \log_2 (A_N(x\mid z)A_N(z\mid y) A_N(y\mid z) A_N(z\mid x))\\
			 &=& J^{\mathrm{num}}(x,z)+J^{\mathrm{num}}(y,z)
		\end{eqnarray*}
		hence the triangle inequality holds.	
	\end{proof}

	The argument in \Cref{jun26-23} has predecessors: for instance, the simple inequality
	\[
		\abs{A\setminus C}\le \abs{A\setminus B}+\abs{B\setminus C}
	\]
	was used in the analysis of the Jaccard distance (see Kjos-Hanssen 2022 \cite{MR4521608}).
	Horibe (1973) \cite{MR345715} gives the following argument which is credited to a conference talk by Shannon (1950) later published in 1953 \cite{1188572}
	(Shannon writes that it is ``readily shown'' but does not give the argument). Let $H(x\mid y)$ denote the entropy of $x$ given $y$. Then
	\begin{eqnarray*}
		H(x\mid z) \le H(x,y \mid z) &=& H(x \mid y,z) + H(y \mid z)\\
								   &\le& H(x \mid y)   + H(y \mid z).
	\end{eqnarray*}
	Similarly, G\'acs, Tromp, and Vit{\'a}nyi \cite{MR1873931} show
	\[
		K(x\mid y^*) \le^+ K(x,z\mid y^*) \le^+ K(z\mid y^*) + K(x \mid z^*).
	\]
	Here, $K$ is the prefix-free Kolmogorov complexity, $K(x\mid y)$ its conditional version; $y^*$ is a shortest program for $y$, so $K(y)=\abs{y^*}$;
	and $a\le^+b$ means $a\le b+O(1)$.
	The heavy lifting was already done in 1974 by G\'acs \cite{MR0403800} who showed ``symmetry of information'', $K(x,y)=^+ K(x) + K(y\mid x^*)$.
	Symmetry of information does not exactly hold in our setting:

	\begin{theorem}
		There exist words $x,y$ with $A_N(x\mid y)A_N(y)\ne A_N(y\mid x)A_N(x)$.
	\end{theorem}
	\begin{proof}
		Let $x=0001$, $y=0011$, then it would imply $2\cdot 3 = 2\cdot 2$.
	\end{proof}

	\begin{remark}\label{exe:max}
		Let $a,b,c,a',b,'c'$ be natural numbers.
		If $a\le bc$ and $a'\le b'c'$ then $\max\{a,a'\}\le\max\{b,b'\}\max\{c,c'\}$,
	\end{remark}

	\begin{theorem}
		The following function $J^{\mathrm{num}}_{\max}$ is a metric on $0\{0,1\}^{n-1}$.
		\begin{eqnarray*}
			J^{\mathrm{num}}_{\max}(x,y) &=& \log_2 \max\{A_N(x\mid y),A_N(y\mid x)\}\\
			&=& \max \{\log_2A_N(x\mid y),\log_2 A_N(y\mid x)\}.
		\end{eqnarray*}
	\end{theorem}
	\begin{proof}
		To show the triangle inequality, applying \Cref{exe:max},
		\[
			\max\{A(x\mid y),A_N(y\mid x)\}\le \max\{A_N(x\mid z),A_N(z\mid x)\}\cdot \max\{A_N(y\mid z),A_N(z\mid y)\}.
		\]
	\end{proof}

	\begin{lemma}[{\cite[Lemma 5]{MR4521608}}]\label{mar29-2022}
		Let $d(x,y)$ be a metric and let $a(x,y)$ be a nonnegative symmetric function. If $a(x,z)\le a(x,y)+d(y,z)$ for all $x,y,z$,
		then $d'(x,y)=\frac{d(x,y)}{a(x,y)+d(x,y)}$,
		with $d'(x,y)=0$ if $d(x,y)=0$, is a metric.
	\end{lemma}

	\begin{theorem}\label{jun20-23}
		The following Jaccard distance type function is a metric on $0\{0,1\}^n$ (with the convention $0/0=0$):
		\[
		J(x,y)=\frac{\log (A_N(x\mid y)A_N(y\mid x))}{\log (A_N(x\mid y)A_N(y\mid x)A_N(x)A_N(y))-\log(A_N(x\# y))}
		\]
	\end{theorem}
	\begin{proof}
		It suffices to prove the triangle inequality.
		We apply \Cref{mar29-2022}.
		Namely, let $d(x,y)=\log_2(A_N(x\mid y)A_N(y\mid x))$ and let $a(x,y)=\log_2\left(\frac{A_N(x)A_N(y)}{A_N(x\# y)}\right)$.
		Then we must show $a(x,z)\le a(x,y)+d(y,z)$ which is equivalent to (writing $A=A_N$ temporarily)
		\[
			A(x\# y) A(x)A(z) \le A(x)A(y)A(y\mid z)A(z\mid y)A(x\# z)
		\]
		\[
			A(x\# y)A(z) \le A(y)A(y\mid z)A(z\mid y)A(x\# z)
		\]
		Since $A(z)\le A(z\mid y)A(y)$, it suffices to show
		\[
			A(x\# y) \le A(y\mid z)A(x\# z)
		\]
		This is shown as follows:
		\[
			A(x\# y) \le A(y\# (x\#z))\le A(y\mid x\# z) A(x\#z) \le A(y\mid z) A(x\#z).
		\]
	\end{proof}
	One advantage of $J$ is that it does not depend on the base of the logarithm chosen.
	\begin{lemma}\label{jac-repeat}
		For the special case where $y=0^n$, a constant word, and $x\ne y$, we have $J(x,y)=1$.
	\end{lemma}
	\begin{proof}
		We compute
		\[
			J(x,y)=\log(A_N(x))/\log(A_N(x)A_N(x)/A_N(x))=1.
		\]
	\end{proof}

	If this metric is very close to the discrete metric, it is of course not very interesting. This is fortunately not the case:

	\begin{theorem}
		There exist words $x,y\in 0\{0,1\}^{n-1}$ with $0<J(x,y)<1/2$, i.e.,
		\[
			A_N(x)A_N(y) \not\le A_N(x\mid y)
			A_N(y\mid x)
			A_N(x\# y).
		\]
	\end{theorem}
	\begin{proof}
		Let $x=u0$, $y=u1$, where $u=0000100$. Then $J(u0,u1)=0.46$.
	\end{proof}

	Calculating $A_N(x\mid y)$ for independent random $x,y$ of lengths $n$ up to 20 we find that the mode of the distribution is around $n/4$ (see Appendix).

	We do not know whether the problem ``''$J(x,y)=1$?'' is decidable in polynomial time.

	\begin{theorem}[Hyde \cite{MR3386523}]\label{hydebound}
		For all words $x$, $A_N(x)\le \abs{x}/2+1$.
	\end{theorem}

	\begin{theorem}\label{jun23-late-23}
		Let $x_1,x_2\in 0\{0,1\}^{n-1}$, $n\le 12$, with $x_1\ne x_2$.
		Let $S=\{x_1,x_2\}$ and let $A=\{001,010,011\}$. Then the following are equivalent:
		\begin{enumerate}
			\item $J(x_1,x_2)=1$;
			\item either
				\begin{enumerate}
					\item $0^n\in S$, or
					\item $n\ge 10$, and $S=\{(01)^{n/2},\alpha^{n/3}\}$ for some $\alpha\in A$.
				\end{enumerate}
		\end{enumerate}
	\end{theorem}
	\begin{proof}
		This is done by computerized search, so we merely give some remarks on the simplifications making the computation feasible.
		If there is an example of length 9 or more then by \Cref{hydebound}, the equation $A_N(x\# y)=A_N(x)A_N(y)$
		must be of the form $a=a\cdot 1$, $a\le 5$ (already ruled out) or $4=2\cdot 2$. So it is enough to check words of complexity 2.
		At length 10 we also have to include the possible equation $6=2\cdot 3$, so we consider all words of complexity 3 or less.
	\end{proof}

	\begin{lemma}\label{jun24-23}
		If $\alpha$ is a permutation word and $k\in\mathbb N$ then $\alpha^k$ is $(k+1)$-powerfree.
	\end{lemma}
	\begin{proof}
		If $w=\alpha^k$ contains a $(k+1)$-power $u$ then let $a$ be the first symbol in $u$.
		Then $a$ appears at least $k+1$ times in $w$. However, there are $\abs{\alpha}$ distinct symbols in $w$ and they each appear
		$k$ times.
	\end{proof}

	\begin{theorem}[{\cite{MR3386523}}]\label{nfa9fact}
		Let $k\in\mathbb N$.
		If an NFA $M$ uniquely accepts $w$ of length $n$, and visits a state $p$ at least $k+1$ times
		during its computation on input $w$,
		then $w$ contains a $k$th power.
	\end{theorem}

	\begin{theorem}\label{jun24-23-sent-men-godt}
		Let $k\in\mathbb N$, $k\ge 1$.
		If a word $w$ is $k$th-powerfree, then $A_N(w)\ge\frac{\abs{w}+1}k$.
	\end{theorem}
	\begin{proof}
		Let $k$ and $w$ be given.
		Let $q=A_N(w)$ and let $M$ witness that $A_N(w)\le q$. For a contrapositive proof, assume $q<\frac{\abs{w}+1}k$.
		Thus $k<\frac{\abs{w}+1}q$.
		
		Let $p$ be a most-visited state in $M$ during its computation on input $w$.
		Then $p$ is visited at least $(\abs{w}+1)/q>k$ times, hence at least $k+1$ times.
		By \Cref{nfa9fact}, $w$ contains an $k$th power. 
	\end{proof}

	\begin{proposition}\label{jun24-23-II}
		For each nonempty permutation word $\alpha$ and each $k\in\mathbb N$ with $k\ge \abs{\alpha}-1$, we have
		$A_N(\alpha^k)=\abs{\alpha}$.
	\end{proposition}
	\begin{proof}
		Let $a=\abs{\alpha}\ge 1$, let $k\ge a-1$, and $w=\alpha^k$. Then by \Cref{jun24-23} and \Cref{jun24-23-sent-men-godt},
		\[
		A_N(w)\ge
		  \frac{\abs{w}+1}{k+1}
		= \frac{ka+1}{k+1}
		= a - \frac{a-1}{k+1} \ge a - \frac{a-1}a > a-1.
		\]
		Since $A_N(w)$ is an integer, $A_N(w)\ge a$. The other direction just uses a single cycle.
	\end{proof}

	From \Cref{jun24-23-II} we have an infinite family of examples with $J(x,y)=1$ as in \Cref{jun23-late-23}.
	Namely, $x$ and $y$ can be powers of permutation words of relatively prime length.

	\begin{lemma}\label{jun24-23-IV}
		If $\alpha$ and $\beta$ are permutation words of lengths $a$ and $b$, then
		\[
			(\alpha^{\mathrm{lcm}(a,b)/a}) \# (\beta^{\mathrm{lcm}(a,b)/b})
		\]
		is a permutation word.
	\end{lemma}
	\begin{proof}
		Let this word $w=w_1\dots w_n$, $w_i\in\Sigma$.
		If $w_i=w_j$ then the first and second coordinates of $w_i$ and $w_j$ are equal, so $i$ is congruent to $j$ mod $a$ and mod $b$.
		Hence $i$ is congruent to $j$ mod $\mathrm{lcm}(a,b)$, so $i=j$.
	\end{proof}

	The case $k=2$, $\alpha=01$, $\beta=012$ exemplifies \Cref{thm:jun24-23}.
	There $a=2$, $b=3$, $u=010101$, $v=012012$, $x=01 01 01 01 01 01$, and $y=012 012 012 012$.
	\begin{theorem}\label{thm:jun24-23}
		Let $\alpha$ and $\beta$ be permutation words of relatively prime lengths $a=\abs{\alpha}$ and $b=\abs{\beta}$.
		Let $k\in\mathbb N$, $k\ge 2$.
		Let $x=u^k$, $y=v^k$, where $u=\alpha^{b}$ and $v=\beta^{a}$. Then $J(x,y)=1$.
	\end{theorem}
	\begin{proof}
		It suffices to show $A_N(x\#y)=A_N(x)A_N(y)$.
		By \Cref{jun24-23-IV}, $u\# v$ is a permutation word, and $x\# y=(u\#v)^k$, and so by \Cref{jun24-23-III},
		\[
			A_N(x\# y)=\abs{u\# v}=\abs{u}=ab=A_N(x)A_N(y).
		\]
	\end{proof}
	In \Cref{jun24-23-III}, the condition $k\ge\abs{\alpha}-1$ in \Cref{jun24-23-II} is strengthened to simply $k\ge 2$ (but of course not to $k\ge 1$).

	\Cref{jun24-23-II} still gives nonredundant information in the case where $2>\abs{\alpha}-1$, i.e., $\alpha\in\{1,2\}$ and $k\in\{0,1\}$.
	\newcommand{\LPI}{n+1-m-\sum_{i=1}^m (\alpha_i-2) \abs{x_i}\le 2q}

	\begin{theorem}[{\cite{MR4313562}}]\label{thm:simple_implies_square_powers_unique_lengths}
		Let $q\ge 1$ and let $x$ be a word such that $A_N(x)\le q$.
		Then $x$ contains a set of
		powers $x_i^{\alpha_i}$, $\alpha_i\ge 2$, $1\le i\le m$ such that
			all the $\abs{x_i}, 1\le i\le m$ are distinct and nonzero,
		and satisfying \eqref{eq:LPI}.
		\begin{equation}\label{eq:LPI}
		\LPI.
		\end{equation}
	\end{theorem}

	\begin{proposition}\label{jun24-23-III}
		For each nonempty permutation word $\alpha$, and $k\ge 2$, we have
		$A_N(\alpha^k)=\abs{\alpha}$.
	\end{proposition}
	\begin{proof}
		First assume $k=2$.
		\Cref{thm:simple_implies_square_powers_unique_lengths} implies that if a square of a permutation word has complexity at most $q$,
		then $n/2\le q$.

		Now let $k>2$. Since $\alpha^2$ is a prefix of $\alpha^k$, we have
		$A_N(\alpha^k)\ge A_N(\alpha^2)=\abs{\alpha}$.
	\end{proof}
	We may wonder whether as long as $\alpha$ is primitive, or at least if $\alpha$ has maximal $A_N$-complexity (achieving Hyde's bound in \Cref{hydebound}), without necessarily being a permutation word,
	\Cref{jun24-23-III} still holds, i.e., $A_N(\alpha^2)=\abs{\alpha}$.
	But this fails:
	\begin{definition}
		A word $w$ has \emph{emergent simplicity} if $A_N(w)$ is maximal, but $A_N(w^2)<\abs{w}$.\index{emergent simplicity}
	\end{definition}
	\begin{proposition}
		There exists a word having emergent simplicity.
		The minimal length of such a word in $\{0,1\}^*$ is 7.
	\end{proposition}
	\begin{proof}
		Let $w=0001000$. Then $w$ is maximally complex, but $A_N(w^2)=6$. The only other example of length 7 is 0010100.
	\end{proof}
	Hence not all maximal-complexity words are like 001, 010, 011 in \Cref{jun23-late-23}.
	That is, in general, maximal-complexity words cannot be used to get instances of $J(x,y)=1$.
	We next show that this emergent simplicity can only go so far, in \Cref{jun25-23}.
	\begin{definition}\index{cyclic shift}\index{conjugate}
		A word $\tilde w$ is a cyclic shift (or a conjugate) of another word $w$ if
		there are words $a,b$ with $w=ab$, $\tilde w=ba$.
	\end{definition}
	\begin{lemma}[{\cite[Theorem 2.4.2]{Shallit:2008:SCF:1434864}}]\label{lem:cyc}
		A cyclic shift of a primitive word is primitive.
	\end{lemma}

	\begin{theorem}[{Lyndon and Sch\"utzenberger \cite{MR0162838}; see \cite[Theorem 2.3.3]{Shallit:2008:SCF:1434864}}]\label{thm:shallit233}\label{lyndon:schuetzenberger}
		Let $x,y\in\Sigma^+$. Then the following four conditions are equivalent:
		\begin{enumerate}
			\item\label{t:1} $xy=yx$.
			\item\label{t:2} There exists $z\in\Sigma$ and integers $k,l>0$ such that $x=z^k$ and $y=z^l$.
			\item\label{t:3} There exist integers $i,j>0$ such that $x^i=y^j$.
		\end{enumerate}
	\end{theorem}

	\begin{theorem}\label{jun25-23}
		If $A_N(w^{\abs{w}})<\abs{w}$ then $w$ is not primitive (and hence does not have maximal $A_N$-complexity).
	\end{theorem}
	\begin{proof}
		If $A_N(w^{\abs{w}})<\abs{w}$ then 
		$A_N(w^{\abs{w}})\le \abs{w}-1$ and hence, since $\abs{w^{\abs{w}}}=\abs{w}^2$ and $k-1\not\ge (k^2+1)/(k+1)$ for all $k$,
		$A_N(w^{\abs{w}})\not\ge \frac{\abs{w}^2+1}{\abs{w}+1}$.
		Therefore, by \Cref{jun24-23-sent-men-godt} $w^{\abs{w}}$ contains an $(\abs{w}+1)$th power $u^{\abs{w}+1}$.
		Thus $(\abs{w}+1)\abs{u}\le \abs{w}^2$, so $\abs{u}<\abs{w}$ (*).
		Since $w^{\abs{w}}$ also contains $u^{\abs{w}}$, there is a cyclic shift $\tilde w$ of $w$ such that $u^{\abs{w}}={\tilde w}^{\abs{u}}$.
		Then letting $g=\mathrm{gcd}(\abs{w},\abs{u})$, we have $u^{\abs{w}/g}={\tilde w}^{\abs{u}/g}$ as well, and now the exponents are relatively prime.
		By \Cref{lyndon:schuetzenberger}, there is a word $\alpha$ with
		$u=\alpha^{\abs{u}/g}$ and $\tilde w=\alpha^{\abs{w}/g}$.
		If $\abs{w}/g>1$ then $\tilde w$ is nonprimitive and hence so is $w$ by \Cref{lem:cyc}.
		If $\abs{w}/g=1$ then $\abs{w}$ divides $\abs{u}$ which is a contradiction to (*).
	\end{proof}
	\begin{theorem}\label{jun23-23}
		Let $x,y\in 0\{0,1\}^{n-1}$. Then $A_N(x\mid y)=1$ iff $x=y$ or $x=0^n$.
	\end{theorem}
	\begin{proof}
		Suppose $x\ne 0^n$ and $M$ is a 1-state NFA such that there is only one labeled walk on which
		there exists $z$ with $M$ accepting $y\#z$ (which we can also write $\binom{y}z$); and on that one walk, $z$ is $x$.
		$M$ has at most four edges, with labels from among $\binom00,\binom01,\binom10,\binom11$.

		Since $x=0u$ and $y=0v$ both start with 0, the walk must start with an edge labeled $\binom00$.
		Then there is no edge labeled $\binom01$, or else $M$ would accept both $y\# x$ and $y\# (1u)$.

		Since $x\ne 0^n$, $x$ contains a 1, so there is an edge labeled $\binom01$ or $\binom11$.
		So there is an edge labeled $\binom11$.

		Then there is no edge labeled $\binom10$, or else two different $z$'s would occur.
		(Let $x=0^k1w$, then $M$ would accept both $y\#x$ and $y\#(0^k0w)$.)
		So $M$ is just the 1-state NFA with two edges labeled $\binom00$ and $\binom11$ only.
		Consequently, $x=y$.
	\end{proof}

This $J$ deems $x$ and $y$ to be ``disjoint'' when $A_N(x\# y)=A_N(x)A_N(y)$,
which for example happens when $x$ and $y$ are high powers of short words of relatively prime length.

\section{A Normalized Information Distance metric}

	The following metric $J_{\max}$, more analogous to the Normalized Information Distance \cite{MR2103495} than the Jaccard distance,
	seems better than $J$ above in the following way:
	If $A_N(x\mid y)=A_N(x)$ and $A_N(y\mid x)=A_N(y)$, then $x$ and $y$ are of no help in compressing each other.
	This should mean that their distance is maximal ($J_{\max}(x,y)=1)$.
	This argument can be compared to that made by Li et al.~\cite{MR2103495}.
	The condition we get from $J$, that $J(x,y)=1$ when $A(x\# y)=A_N(x)A_N(y)$, seems relatively unmotivated in comparison.

	\begin{definition}
		Let $x,y$ be words of length $n\in\mathbb N$. We define
		\[
			J_{\max}(x,y) = \frac{\log \max\{A(x\mid y),A(y\mid x)\}}{\log \max\{A(x),A(y)\}}
		\]	
	\end{definition}

	\begin{theorem}\label{appthm}
		$J_{\max}$ is a metric on the set $0\{0,1\}^{n-1}$.
	\end{theorem}
	The proof is similar to that of \Cref{jun20-23} and is given in the Appendix.
	For this metric the case $J(x,y)=1$ is perhaps not much easier to understand.
	We already have examples where $A(y\mid x)=A(y)$ even though $x$ and $y$ are both nontrivial, such as $x,y\in\{001,010,011\}$.
	We can ask whether this metric embeds in Euclidean space but it fails already at $n=4$.

	\paragraph{Conclusion and future work.}
	We have seen that while automatic complexity $A_N$ is quite different from Kolmogorov complexity $K$,
	surprisingly we can obtain metrics similar to those for $K$ using $\log A_N$ and a conditional version of $A_N$.
	In fact, these metrics are genuine metrics (not just up to some accuracy) and are computable.
	We also saw that the conditional version of $A_N$ sheds more light upon the problem of distinguishing $A_{N}=A_{Nu}$ and its word-counting version $A_{Ne}$. 
	In the future, we hope to characterize when $J=1$ and $J_{\max}=1$, and determine whether $A_N$, $J$ or $J_{\max}$ are efficiently computable.

	\bibliographystyle{splncs04}
	\bibliography{cocoon23}
\section*{Appendix}
	The distributions of $q=A_{Nu}(x\mid y)$ for $n\le 10$ are as follows (modes indicated in brackets):

	\begin{tabular}{l l l l l l l}
		$n\setminus q$ &1 &2 &3 &4 &5 &6\\
		0	&  [1]\\
		1	&  [1]\\
		2	&  [3]			& 1\\
		3	&  7			& [9]\\
		4	& 15			& [45]	& 4\\
		5	& 31			& [197]	& 28\\
		6	& 63			& [755]	& 191		& 15\\
		7	& 127			& [2299]& 1561		& 109\\
		8	& 255			& 5905	& [9604]	& 571	& 49\\
		9	& 511			& 14005	& [47416]	& 3205	& 399\\
		10	& 1023			& 31439	& [206342]	& 21066	& 2102	& 172\\
	\end{tabular}

	For example, the table entry for $(n,q)=(2,2)$ is 1 since the only instance of $x,y\in 0\{0,1\}$ with $A_N(x\mid y)=2$ is $A_N(01\mid 00)$.

\paragraph{Proof of \Cref{appthm}.}
	\begin{proof}
		The nontrivial part is the triangle inequality.
		Let
		\[a_{xy} = \log \max \{A(x),A(y)\} - \log \max \{A(x\mid y),A(y\mid x)\};\]
		then by \Cref{mar29-2022} it suffices to show that
		$a_{xz} \le a_{xy} + d_{yz}$. In other words:
		\begin{eqnarray*}
			&&	 \log \max \{A(x),A(z)\} - \log \max\{ A(x\mid z),A(z\mid x)\}\\
			&\le&\log \max\{ A(x),A(y)\} - \log \max \{A(x\mid y),A(y\mid x)\}+
				 \log \max \{A(y),A(z)\}.
		\end{eqnarray*}

		Equivalently, we must show that
		\begin{eqnarray*}
			&&	\max \{A(x\mid y),A(y\mid x)\}  \max \{A(x),A(z)\}\\
			&\le&
			\max \{A(x\mid z),A(z\mid x)\} \max \{A(y),A(z)\}  \max \{A(x),A(y)\}.
		\end{eqnarray*}
		 There are now two cases.

		\noindent \emph{Case 1: $A(x)\le A(z)$.} Then in fact
		\begin{eqnarray*}
			&&	\max \{A(x\mid y),A(y\mid x)\}  \max \{A(x),A(z)\}\\
			&\le&
			\max \{A(y),A(z)\}  \max \{A(x),A(y)\}.
		\end{eqnarray*}
		\noindent\emph{Case 2: $A(z)<A(x)$.} Then we use
		$A(x\mid y)\le A(x)\le A(x\mid z)A(z)$
		and
		$A(y\mid x)\le A(y)$.
	\end{proof}

\end{document}